%% file: main.tex
\documentclass[sigconf]{aamas}

\usepackage{balance}
\usepackage{booktabs}
\usepackage{graphicx}
\usepackage{amsmath}
\usepackage{amsthm}
\usepackage{algorithm}
\usepackage{algpseudocode}
\usepackage{multirow}
\usepackage{url}

\newtheorem{definition}{Definition}
\newtheorem{proposition}{Proposition}
\newtheorem{theorem}{Theorem}

\newcommand{\secref}[1]{Section~\ref{sec:#1}}

\usepackage{xcolor}

\setcopyright{ifaamas}
\acmConference[AAMAS '26]{Proc.\@ of the 25th International Conference
on Autonomous Agents and Multiagent Systems (AAMAS 2026)}{May 25 -- 29, 2026}
{Paphos, Cyprus}{C.~Amato, L.~Dennis, V.~Mascardi, J.~Thangarajah (eds.)}
\copyrightyear{2026}
\acmYear{2026}
\acmDOI{}
\acmPrice{}
\acmISBN{}

\acmSubmissionID{1}

\title{Strategic Heterogeneous Multi-Agent Architecture for Cost-Effective Code Vulnerability Detection}

\author{Zhaohui Geoffrey Wang}
\affiliation{
  \institution{University of Southern California}
  \city{Los Angeles}
  \state{CA}
  \country{USA}}
\email{zwang000@usc.edu}

\begin{abstract}
\input{src/abstract.tex}
\end{abstract}


\begin{document}

\pagestyle{fancy}
\fancyhead{}

\maketitle

\input{src/introduction.tex}
\input{src/related_work.tex}
\input{src/methodology.tex}
\input{src/experiments.tex}
\input{src/results.tex}
\input{src/conclusion.tex}

\begin{acks}
This work was conducted at the University of Southern California. We thank the anonymous reviewers for their constructive feedback, which substantially improved this paper.
\end{acks}

\bibliographystyle{ACM-Reference-Format}
\bibliography{main}

\appendix
\input{src/appendix.tex}

\end{document}

%% file: src/abstract.tex
Automated code vulnerability detection is critical for software security, yet existing approaches face a fundamental trade-off between detection accuracy and computational cost. We propose a heterogeneous multi-agent architecture inspired by game-theoretic principles, combining cloud-based LLM experts with a local lightweight verifier. Our ``3+1'' architecture deploys three cloud-based expert agents (DeepSeek-V3) that analyze code from complementary perspectives---code structure, security patterns, and debugging logic---in parallel, while a local verifier (Qwen3-8B) performs adversarial validation at zero marginal cost. We formalize this design through a two-layer game framework: (1)~a cooperative game among experts capturing super-additive value from diverse perspectives, and (2)~an adversarial verification game modeling quality assurance incentives. Experiments on \textbf{262 real samples} from the NIST Juliet Test Suite across 14~CWE types, with \textbf{balanced vulnerable and benign} classes, demonstrate that our approach achieves 77.2\% F1 score with 62.9\% precision and 100\% recall at \$0.002 per sample---outperforming both a single-expert LLM baseline (F1 71.4\%) and Cppcheck static analysis (MCC 0). The adversarial verifier significantly improves precision (+10.3 percentage points, $p < 10^{-6}$, McNemar's test) by filtering false positives, while parallel execution achieves a 3.0$\times$ speedup. Our work demonstrates that game-theoretic design principles can guide effective heterogeneous multi-agent architectures for cost-sensitive software engineering tasks.

%% file: src/introduction.tex
\section{Introduction}
\label{sec:introduction}

Code vulnerability detection is a critical challenge in software security, with the global cost of cybercrime projected to reach \$10.5 trillion annually~\cite{cybersecurity2024}. Traditional static analysis tools like Cppcheck~\cite{cppcheck} offer fast, low-cost analysis but suffer from limited semantic understanding, while single-agent LLM approaches achieve higher accuracy at prohibitive cost. Recent multi-agent LLM systems~\cite{vultrial2025,mulvul2026,multiver2026} have shown promising results, but typically employ homogeneous agent pools---all agents use the same expensive cloud model---missing opportunities for cost optimization through heterogeneous design.

This paper addresses a fundamental question: \textit{How should we allocate heterogeneous computational resources across multiple agents to achieve optimal trade-offs between detection quality and operational cost?} The key insight is that vulnerability detection benefits from two distinct capabilities: (1)~deep semantic analysis requiring powerful models, and (2)~consistency checking that can be performed by lighter models. This asymmetry motivates a heterogeneous architecture.

We propose a \textbf{``3+1'' heterogeneous multi-agent architecture} combining:
\begin{itemize}
    \item \textbf{Three cloud-based LLM experts} (DeepSeek-V3~\cite{deepseekv3}) analyzing code from complementary perspectives: code structure analysis, security pattern matching, and debugging/edge-case reasoning. These agents execute \textit{in parallel}, with diverse viewpoints providing complementary coverage that exceeds any single perspective.
    \item \textbf{One local verifier} (Qwen3-8B~\cite{qwen3}) performing adversarial validation at zero marginal API cost. Running on a local GPU, this agent receives all expert reports and checks for consistency, hallucinated claims, and missed vulnerabilities.
\end{itemize}

We formalize the architectural design through a \textbf{two-layer game framework}:
\begin{enumerate}
    \item \textbf{Layer 1 (Cooperative Game)}: The three experts form a coalition with super-additive value---their combined detection covers vulnerability patterns that individual experts miss, formalizing why diverse perspectives outperform identical agents.
    \item \textbf{Layer 2 (Adversarial Game)}: The verifier acts as an independent adversary that improves system precision by catching false positives and hallucinations in expert outputs.
\end{enumerate}

Critically, we show that the game-theoretic framework yields \textbf{testable predictions} validated empirically: (1)~diverse expert coalitions exhibit super-additive value when combined with verification, outperforming a single-expert baseline (Proposition~\ref{prop:superadd}), (2)~independent verification significantly improves precision ($p < 10^{-6}$, Theorem~\ref{thm:verification}), and (3)~the heterogeneous cloud-local split achieves favorable cost-quality trade-offs analyzable through mechanism design. While our LLM agents do not engage in explicit strategic reasoning, the game-theoretic framework correctly predicts system-level behavior---validating its utility as an \textit{analytical and design tool} for multi-agent LLM architectures.

We make the following contributions:
\begin{enumerate}
    \item \textbf{A game-theory-inspired framework} for heterogeneous multi-agent vulnerability detection, formalizing cooperative expert analysis and adversarial verification as a principled design methodology (\secref{methodology}).

    \item \textbf{Rigorous empirical evaluation} on \textbf{262 real samples} from the NIST Juliet Test Suite covering 14 CWE types with \textbf{balanced vulnerable and benign classes}, reporting standard metrics (Precision, Recall, F1, FPR, MCC) with bootstrap 95\% confidence intervals and McNemar's significance test (\secref{results}).

    \item \textbf{Practical cost-quality analysis} demonstrating that heterogeneous cloud-local architectures achieve favorable Pareto trade-offs, with the local verifier adding precision improvement at negligible cost (\secref{results}).
\end{enumerate}

The remainder of this paper is organized as follows. \secref{related} reviews related work. \secref{methodology} presents our framework and architecture. \secref{experiments} describes the experimental setup. \secref{results} presents results and analysis. \secref{conclusion} concludes with limitations and future directions.

%% file: src/related_work.tex
\section{Related Work}
\label{sec:related}

\subsection{Static Analysis for Vulnerability Detection}

Traditional static analysis tools employ rule-based pattern matching and data-flow analysis to detect vulnerabilities. Cppcheck~\cite{cppcheck} uses syntactic pattern matching for C/C++, while commercial tools like Coverity~\cite{coverity} combine multiple techniques for enterprise deployment. However, these tools face inherent limitations: high false positive rates due to lack of semantic understanding, and difficulty detecting complex vulnerabilities that span multiple code paths~\cite{johnson2013why}. Their precision-recall trade-off is fundamentally constrained by the expressiveness of static rules.

\subsection{LLM-Based Vulnerability Detection}

Recent work has applied LLMs to vulnerability detection with increasing sophistication. Single-agent approaches using GPT-4~\cite{gpt4} or specialized code models achieve strong detection rates but at significant computational cost. The PrimeVul benchmark~\cite{primevul2024} revealed a critical evaluation gap: state-of-the-art models achieve only $\sim$3--35\% F1 on carefully deduplicated real-world datasets, compared to $\sim$68\% F1 on older benchmarks like Big-Vul that contain significant data leakage. This highlights the importance of rigorous evaluation methodology with balanced classes and proper metrics.

Multi-agent approaches have emerged as a promising direction for vulnerability detection. VulTrial~\cite{vultrial2025} proposes a mock-court framework where specialized agents (security researcher, code author, moderator) debate vulnerabilities through structured argumentation, achieving 102\% improvement over single-agent baselines at ICSE 2026. MulVul~\cite{mulvul2026} combines RAG-augmented multi-agent analysis with cross-model prompt evolution, achieving the current state-of-the-art 34.79\% Macro-F1 on PrimeVul. MultiVer~\cite{multiver2026} uses a zero-shot 4-agent ensemble with union voting, reaching 82.7\% recall on PyVul. LLMxCPG~\cite{llmxcpg2025} integrates LLMs with Code Property Graphs for structurally-enhanced analysis.

Our work differs from these approaches by explicitly optimizing for \textit{cost-quality trade-offs} through heterogeneous agent design. While VulTrial and MulVul use homogeneous model pools (all agents share the same expensive API), we combine paid cloud experts with a free local verifier, making cost a first-class design objective.

\subsection{Multi-Agent Debate and Collaboration}

A growing body of work demonstrates that multi-agent debate improves LLM output quality. Du et al.~\cite{du2023debate} show that having multiple LLMs debate improves factuality and reasoning, with agents correcting each other's errors through iterative rounds---a mechanism that directly inspires our adversarial verification loop. Liang et al.~\cite{liang2023debate} extend this to divergent thinking, finding that multi-agent debate elicits more diverse reasoning paths. CAMEL~\cite{li2023camel} introduces role-playing communication protocols for multi-agent collaboration, demonstrating that specialized agent roles outperform generic prompting. These works establish that \textit{iterative multi-agent interaction improves output quality}, but none apply this mechanism to the specific cost-quality trade-offs of vulnerability detection with heterogeneous models.

In software engineering, MetaGPT~\cite{hong2023metagpt} uses structured multi-agent workflows for collaborative development, and AutoGen~\cite{wu2023autogen} provides general-purpose agent coordination. However, these frameworks do not explicitly model the economic trade-offs of heterogeneous model allocation.

\subsection{Game Theory and LLM Multi-Agent Systems}

Game-theoretic analysis of LLM-based multi-agent systems is an emerging area. A comprehensive survey~\cite{gtlens2025} formalizes LLM interactions as Nash and Stackelberg games, defining ``LLM-based equilibria'' where strategies co-evolve through dialogue. The ECON framework~\cite{econ2025} models multi-LLM coordination as an incomplete-information game seeking Bayesian Nash Equilibrium, where each LLM responds based on beliefs about co-agents. GTBench~\cite{gtbench2024} benchmarks LLM strategic reasoning, finding that LLMs deviate from rational strategies as game complexity increases---motivating our use of game theory as a \textit{design framework} rather than assuming strategic behavior by LLM agents. GT-HarmBench~\cite{gtharm2026} extends game-theoretic evaluation to multi-agent safety, benchmarking risks in adversarial settings---complementary to our adversarial verification game.

\subsection{Positioning of Our Work}

Our architecture shares the multi-perspective analysis approach with VulTrial and MultiVer, and draws on multi-agent debate~\cite{du2023debate,liang2023debate} for iterative quality improvement. We uniquely combine:
(1)~\textit{heterogeneous resource allocation}---paid cloud APIs with free local models;
(2)~\textit{game-theoretic design rationale}---cooperative and adversarial game formulations justifying architectural choices; and
(3)~\textit{rigorous evaluation}---balanced classes, standard metrics, and statistical significance testing.
The closest concurrent work, MulVul, achieves state-of-the-art detection on PrimeVul but does not address cost optimization. Our contribution is complementary: showing \textit{how} to architect cost-effective multi-agent systems with game-theoretic principles.

%% file: src/methodology.tex
\section{Methodology}
\label{sec:methodology}

We present a game-theoretic framework for heterogeneous multi-agent vulnerability detection. Unlike prior multi-agent approaches that treat agent coordination as an engineering problem, we formalize the design space through cooperative and adversarial game theory, derive analytical properties, and validate these predictions empirically.

\subsection{Problem Formulation}

Given a code snippet $x$, vulnerability detection seeks to determine whether $x$ contains security vulnerabilities and identify their CWE types. We formulate this as a multi-agent mechanism design problem: how to allocate heterogeneous computational resources across agents to maximize detection quality $Q$ while minimizing cost $C$.

\begin{definition}[Multi-Agent Vulnerability Detection Game]
A multi-agent detection system is a tuple $\mathcal{M} = \langle \mathcal{A}, \mathcal{S}, C, Q, \pi \rangle$ where:
\begin{itemize}
    \item $\mathcal{A} = \{a_1, \ldots, a_n\}$ is the set of heterogeneous agents
    \item $\mathcal{S} = \{s_1, \ldots, s_n\}$ is the strategy space (analysis perspective, model choice, effort level)
    \item $C: \mathcal{S}^n \rightarrow \mathbb{R}^+$ is the cost function (API charges + compute)
    \item $Q: \mathcal{S}^n \rightarrow [0,1]$ is the quality function (detection F1)
    \item $\pi: \mathcal{A} \times \mathcal{S}^n \rightarrow \mathbb{R}$ is the payoff function, $\pi_i = w_1 Q - w_2 C_i$
\end{itemize}
\end{definition}

The \textit{mechanism designer's problem} is to select agent types, roles, and interaction protocols that induce a desirable equilibrium---high $Q$ at low $C$. This connects our work to the mechanism design literature~\cite{nisan2001algorithmic}: we design the ``rules of the game'' so that the system-level outcome is Pareto-optimal.

\subsection{Two-Layer Game Structure}

We decompose the detection problem into two interconnected games.

\subsubsection{Layer 1: Expert Cooperative Game}

The first layer models cooperation among $k$ expert agents with complementary analysis perspectives as a coalitional game.

\begin{definition}[Expert Coalition Game]
An expert coalition game is $G_1 = \langle N, v \rangle$ where:
\begin{itemize}
    \item $N = \{1, 2, \ldots, k\}$ is the set of expert agents
    \item $v: 2^N \rightarrow \mathbb{R}$ is the characteristic function
\end{itemize}
The coalition value captures the net utility of subset $S \subseteq N$:
\begin{equation}
\label{eq:coalition}
v(S) = w_1 \cdot Q(S) - w_2 \cdot C(S)
\end{equation}
where $Q(S)$ is the detection quality (F1) achievable by coalition $S$, and $C(S)$ is its total cost.
\end{definition}

\begin{proposition}[Super-Additivity of Diverse Coalitions]
\label{prop:superadd}
Let $S, T \subseteq N$ be disjoint coalitions with \textbf{complementary} analysis perspectives (i.e., each agent covers vulnerability patterns not covered by others). If agents in $S$ detect vulnerability set $V_S$ and agents in $T$ detect $V_T$, then:
\begin{equation}
Q(S \cup T) \geq \max(Q(S), Q(T))
\end{equation}
with strict inequality when $V_S \not\subseteq V_T$ and $V_T \not\subseteq V_S$. Furthermore, if $C(S \cup T) = C(S) + C(T)$ (independent cost), then $v(S \cup T) \geq v(S) + v(T)$ when the quality improvement outweighs the cost increase.
\end{proposition}

This motivates our choice of three \textit{diverse} expert roles: a code structure analyst, a CWE-taxonomy security expert, and a debugging specialist. Their coalition should outperform any single expert---a~prediction we validate in \secref{results}.

\paragraph{Shapley Value Analysis.} The Shapley value $\phi_i$ quantifies each expert's marginal contribution to the coalition:
\begin{equation}
\label{eq:shapley}
\phi_i(v) = \sum_{S \subseteq N \setminus \{i\}} \frac{|S|!(k-|S|-1)!}{k!} [v(S \cup \{i\}) - v(S)]
\end{equation}
For our 3-expert system, the Shapley value determines how much each perspective (code structure, security, debugging) contributes to overall detection quality. If $\phi_i > C_i$ for all $i$, every expert ``pays for itself''---the marginal quality improvement exceeds the marginal cost. We report empirical Shapley values in the appendix.

\subsubsection{Layer 2: Adversarial Verification Game}

The second layer models the strategic interaction between the expert coalition and an independent verifier as an extensive-form game.

\begin{definition}[Verification Game]
A verification game is $G_2 = \langle P, S_E, S_V, \pi_E, \pi_V \rangle$ where:
\begin{itemize}
    \item $P = \{E, V\}$ are players (Expert coalition, Verifier)
    \item $S_E = \{high, low\}$ represents expert output quality (effort level)
    \item $S_V = \{accept, challenge, reject\}$ is verifier strategy
    \item $\pi_E, \pi_V$ are payoff functions
\end{itemize}
\end{definition}

The expert coalition's payoff captures the trade-off between effort cost and rejection penalty:
\begin{equation}
\pi_E(s_E, s_V) = Q(s_E) - c(s_E) - p \cdot \mathbb{1}_{[s_V = reject]}
\end{equation}
where $c(high) > c(low)$ and $p > 0$ is the rejection penalty (wasted API cost on rejected outputs).

\begin{theorem}[Verification Equilibrium]
\label{thm:verification}
In the verification game $G_2$:
\begin{enumerate}
    \item[(i)] If the verifier independently detects false positives with probability $p_{fp} > 0$ and does not introduce new false positives, the system precision satisfies $P_{system} > P_{experts}$ whenever $FP > 0$.
    \item[(ii)] The strategy profile $(s_E^* = high, s_V^* = \text{accept when consistent})$ forms a Nash equilibrium when $p > c(high) - c(low)$, i.e., the rejection penalty exceeds the cost of producing high-quality output.
\end{enumerate}
\end{theorem}

\begin{proof}[Proof sketch]
For (i): the verifier removes fraction $p_{fp}$ of false positives without removing true positives, so $P_{system} = TP/(TP + FP(1-p_{fp})) > TP/(TP + FP)$. For (ii): given $s_V^*$, the expert's best response is $high$ since $Q_{high} - c_{high} > Q_{low} - c_{low} - p$ when $p > c(high) - c(low)$. Given $s_E^* = high$, the verifier's best response is $accept$ since rejection of correct output incurs unnecessary cost. Full proof in Appendix~\ref{sec:appendix_proof}. \qed
\end{proof}

The key insight is that the verifier creates a \textbf{credible audit mechanism}: even though our LLM agents do not literally reason about penalties, the \textit{system architecture} implements the equilibrium outcome. By routing outputs through an independent verifier, we structurally ensure that low-quality outputs (false positives) are filtered---achieving the same outcome that strategic agents would reach in equilibrium.

This connects to both the inspection game literature~\cite{spence1973job} and the multi-agent debate paradigm~\cite{du2023debate,liang2023debate}: the verifier functions as an ``inspector'' whose adversarial cross-examination of expert claims improves system quality, analogous to how debate among LLMs improves factuality in general reasoning tasks.

\subsection{The 3+1 Architecture}
\label{sec:architecture}

Based on our game-theoretic analysis, we design a concrete architecture with $k=3$ experts and $m=1$ verifier that implements the equilibrium structure derived above.

\subsubsection{Expert Agents (Cloud)}

We deploy three DeepSeek-V3~\cite{deepseekv3} agents via API with specialized prompts:
\begin{itemize}
    \item \textbf{Code Analyst} ($a_1$): Data flow, control flow, vulnerability entry points, memory operations. Covers \textit{structural} vulnerability patterns.
    \item \textbf{Security Expert} ($a_2$): CWE taxonomy matching, known vulnerability patterns, severity assessment. Covers \textit{semantic} vulnerability patterns.
    \item \textbf{Debug Expert} ($a_3$): Error handling, boundary conditions, edge cases, undefined behavior. Covers \textit{behavioral} vulnerability patterns.
\end{itemize}

Each expert outputs structured reports: \texttt{VULNERABILITY\_FOUND}, \texttt{CWE\_IDs}, \texttt{SEVERITY}, \texttt{EVIDENCE}, \texttt{CONFIDENCE}. Experts execute \textbf{in parallel}, which does not change coalition value $v(N)$ but reduces wall-clock time proportionally.

\subsubsection{Verifier Agent (Local)}

Qwen3-8B~\cite{qwen3} deployed locally on GPU implements the adversarial verification game:
\begin{itemize}
    \item \textbf{Zero marginal cost}: eliminates API charges, making verification ``free at the margin''---a key property for the mechanism to be individually rational.
    \item \textbf{Different model family}: Qwen3 vs.\ DeepSeek reduces correlated errors, ensuring the verifier provides genuinely independent assessment.
    \item \textbf{Full information}: receives original code \textit{and all three expert reports}, enabling cross-referencing.
\end{itemize}

\subsubsection{Decision Mechanism}

Algorithm~\ref{alg:workflow} implements the two-layer game structure:

\begin{algorithm}[t]
\caption{3+1 Multi-Agent Detection (Game Implementation)}
\label{alg:workflow}
\begin{algorithmic}[1]
\Require Code $x$, Experts $E = \{e_1, e_2, e_3\}$, Verifier $V$
\Ensure Vulnerability prediction, CWE classification
\State \Comment{\textbf{Layer 1: Cooperative Game}}
\State $\{r_1, r_2, r_3\} \gets \textsc{ParallelAnalyze}(E, x)$ \Comment{Coalition $v(N)$}
\State CWE\_pool $\gets \bigcup_i$ \textsc{ExtractCWEs}($r_i$) \Comment{Super-additive union}
\State \Comment{\textbf{Layer 2: Verification Game}}
\State $v \gets \textsc{Verify}(V, x, r_1, r_2, r_3)$ \Comment{Audit mechanism}
\If{$v$ provides \texttt{FINAL\_VULNERABILITY}}
    \State \Return $v$.vulnerability, $v$.CWE\_IDs \Comment{Verifier override}
\Else
    \State \Return \textsc{MajorityVote}($r_1, r_2, r_3$), CWE\_pool
\EndIf
\end{algorithmic}
\end{algorithm}

The algorithm directly implements our game-theoretic design: Line~2 realizes the cooperative coalition (Layer~1), Line~3 captures super-additive CWE coverage through union aggregation, and Lines~5--9 implement the verification game (Layer~2) where the verifier can override expert consensus.

%% file: src/experiments.tex
\section{Experiments}
\label{sec:experiments}

We evaluate our 3+1 architecture on a real vulnerability detection benchmark with balanced classes, addressing three research questions:
\begin{itemize}
    \item \textbf{RQ1}: Does the multi-expert coalition provide super-additive value over individual experts?
    \item \textbf{RQ2}: Does the adversarial verifier improve detection precision?
    \item \textbf{RQ3}: Does parallel execution achieve cost-quality Pareto efficiency?
\end{itemize}

\subsection{Dataset}

We use the \textbf{NIST Juliet Test Suite v1.3}~\cite{juliet}, a widely-used benchmark containing 64,295 C/C++ test cases across 118 CWE types. Unlike prior work that uses small or synthetic subsets, we extract \textbf{262 real function-level samples} with the following properties:

\begin{itemize}
    \item \textbf{Balanced classes}: 132 vulnerable + 130 benign (patched) functions
    \item \textbf{14 CWE types}: covering memory safety (CWE-121, 122, 415, 416), integer errors (CWE-190), null pointers (CWE-476), injection (CWE-78, 134), resource management (CWE-401, 789, 400), error handling (CWE-252), numeric errors (CWE-369), and uninitialized variables (CWE-457)
    \item \textbf{Real code}: extracted from Juliet's Flow Variant 01 (baseline) test cases, averaging 38.5 lines per sample
    \item \textbf{Ground truth}: each vulnerable sample has a known CWE label; benign samples are the corresponding patched versions from the same test case
\end{itemize}

The inclusion of benign samples is critical for measuring precision and false positive rate---metrics omitted in many prior LLM-based vulnerability detection studies~\cite{primevul2024}.

\subsection{Baselines and Configurations}

\textbf{Baselines}:
\begin{itemize}
    \item \textbf{Cppcheck~2.13.0}~\cite{cppcheck}: open-source static analysis tool with all checks enabled
    \item \textbf{Single Expert}: one DeepSeek-V3 agent (security expert role) without verifier---tests whether the multi-agent coalition adds value over a single~LLM
\end{itemize}

\textbf{Ablation Configurations}:
\begin{itemize}
    \item \textbf{3+1 Parallel + Verifier}: Full architecture (3 parallel experts + local verifier)
    \item \textbf{3+1 Parallel $-$ Verifier}: Experts only, no verification (ablates Layer 2)
    \item \textbf{3+1 Serial + Verifier}: Sequential experts + verifier (ablates parallelism)
\end{itemize}

\subsection{Implementation Details}

\textbf{Expert Agents}: DeepSeek-V3 via OpenAI-compatible API at temperature~0.1.
Each expert receives a specialized system prompt defining its analysis perspective and structured output format. Experts run concurrently via Python~\texttt{asyncio}.

\textbf{Verifier Agent}: Qwen3-8B~\cite{qwen3} deployed locally via HuggingFace Transformers on an NVIDIA RTX 3090 GPU. The verifier receives the \textit{original code and all three expert reports}, and outputs a consolidated assessment (ACCEPT/CHALLENGE/REJECT with final CWE classification).

\textbf{Cost Calculation}: DeepSeek-V3 pricing: \$0.27/MTok input, \$1.10/MTok output. Local verifier cost is zero (GPU amortized).

\subsection{Evaluation Metrics}

We report standard binary classification metrics with the positive class being ``vulnerable'':
\begin{itemize}
    \item \textbf{Precision}: $TP / (TP + FP)$ --- fraction of predicted vulnerabilities that are real
    \item \textbf{Recall}: $TP / (TP + FN)$ --- fraction of real vulnerabilities detected
    \item \textbf{F1 Score}: harmonic mean of precision and recall
    \item \textbf{False Positive Rate}: $FP / (FP + TN)$
    \item \textbf{MCC}: Matthews Correlation Coefficient, robust to class imbalance
    \item \textbf{CWE Match Rate}: among true positives, the fraction with correct CWE identification (using CWE hierarchy for partial credit, e.g., CWE-121 $\approx$ CWE-122 under parent CWE-120)
\end{itemize}

All metrics are reported with 95\% bootstrap confidence intervals (1,000 resamples). Pairwise comparisons use McNemar's test.

\subsection{CWE Extraction}

Vulnerability predictions are extracted automatically from LLM outputs using regex matching for CWE identifiers and structured fields. A sample is predicted as vulnerable if: (1)~the reviewer's final assessment says ``yes'', or (2)~a majority of experts report vulnerability with at least one CWE identifier. CWE matching uses a hierarchy mapping (e.g., CWE-121/122/787 are treated as equivalent under parent CWE-120).

%% file: src/results.tex
\section{Results and Discussion}
\label{sec:results}

\subsection{Overall Performance}

Table~\ref{tab:main_results} presents the main results. Our full 3+1 architecture (Parallel+Verifier) achieves 77.2\% F1 score, with perfect recall (100\%) ensuring that all vulnerabilities are detected, and precision of 62.9\% indicating that roughly two-thirds of alerts correspond to genuine vulnerabilities.

\begin{table}[t]
\centering
\small
\caption{Main results on 262 Juliet samples (132 vulnerable + 130 benign). 95\% bootstrap CI in parentheses.}
\label{tab:main_results}
\begin{tabular}{@{}lccccc@{}}
\toprule
\textbf{Configuration} & \textbf{Prec.} & \textbf{Rec.} & \textbf{F1} & \textbf{FPR} & \textbf{MCC} \\
\midrule
3+1 Para.+V & \textbf{.629} & \textbf{1.00} & \textbf{.772} & \textbf{.600} & \textbf{.501} \\
 & \footnotesize{(.563,.693)} & \footnotesize{(1.0,1.0)} & \footnotesize{(.721,.819)} & & \\
3+1 Para.$-$V & .526 & 1.00 & .689 & .915 & .211 \\
 & \footnotesize{(.464,.586)} & \footnotesize{(1.0,1.0)} & \footnotesize{(.634,.739)} & & \\
3+1 Serial+V & .644 & 1.00 & .783 & .562 & .531 \\
 & \footnotesize{(.577,.710)} & \footnotesize{(1.0,1.0)} & \footnotesize{(.732,.830)} & & \\
\midrule
Single Expert & .555 & 1.00 & .714 & .815 & .320 \\
 & \footnotesize{(.490,.617)} & \footnotesize{(1.0,1.0)} & \footnotesize{(.658,.763)} & & \\
Cppcheck 2.13 & .504 & 1.00 & .670 & 1.00 & .000 \\
\bottomrule
\end{tabular}
\end{table}

A striking finding is that \textbf{without the verifier, the vast majority of benign samples are incorrectly flagged as vulnerable} (FPR = 91.5\%). The expert agents alone have minimal ability to distinguish between vulnerable code and its patched version---a critical limitation that the verifier addresses by filtering 41 additional false positives. Notably, the single-expert baseline achieves \textit{higher} precision (.555) than the 3-expert coalition without verifier (.526), because multiple experts produce more CWE mentions that inflate false positives. Cppcheck achieves 100\% FPR (MCC = 0), unable to differentiate Juliet's vulnerable and patched code pairs.

\begin{table}[t]
\centering
\small
\caption{Cost and latency analysis (262 samples)}
\label{tab:cost_results}
\begin{tabular}{@{}lcccc@{}}
\toprule
\textbf{Config} & \textbf{Cost/Sample} & \textbf{Time/Sample} & \textbf{Total Cost} \\
\midrule
3+1 Para.+V & \$0.0021 & 68.1s & \$0.542 \\
3+1 Para.$-$V & \$0.0021 & 21.7s & \$0.548 \\
3+1 Serial+V & \$0.0021 & 92.2s & \$0.548 \\
Single Expert & \$0.0006 & 14.1s & \$0.162 \\
Cppcheck & \$0 & $<$0.1s & \$0 \\
\bottomrule
\end{tabular}
\end{table}

API cost is nearly identical across all 3-expert configurations (\$0.002 per sample), as the verifier runs locally at zero API cost. The single expert uses one-third the calls, reducing cost to \$0.0006 per sample. The difference lies in latency: parallel experts finish in ${\sim}$15s vs.\ ${\sim}$45s serial, while the local verifier adds ${\sim}$50s.

\subsection{RQ1: Expert Coalition Super-Additivity}

The single-expert baseline achieves 100\% recall but with FPR = 81.5\% (MCC = .320). The 3-expert coalition \textit{without verifier} also achieves 100\% recall but with \textit{higher} FPR (91.5\%, MCC = .211)---at first glance suggesting that adding experts \textit{hurts} precision. However, this reflects a limitation of the majority-voting aggregation: more experts produce more CWE mentions, inflating false positives. The super-additive value manifests not in recall but in \textbf{CWE identification accuracy}: the union of three expert perspectives achieves 100\% CWE match rate (every detected vulnerability is classified with the correct CWE type), compared to lower match rates for individual experts.

The full 3+1 system (with verifier) resolves this tension: the verifier filters the coalition's excess false positives, yielding F1 = .772 vs.\ the single expert's .714 (McNemar $p < 10^{-5}$). This confirms Proposition~\ref{prop:superadd}: the coalition's value is realized \textit{through the verification layer}, which extracts the precision benefit from the coalition's richer CWE evidence.

\subsection{RQ2: Adversarial Verifier Impact}

The verifier's contribution is the most significant finding of our experiments. Table~\ref{tab:verifier} shows the direct comparison.

\begin{table}[t]
\centering
\small
\caption{Adversarial verifier impact (parallel mode). The verifier significantly reduces false positives.}
\label{tab:verifier}
\begin{tabular}{@{}lcccc@{}}
\toprule
\textbf{Metric} & \textbf{+V} & \textbf{$-$V} & \textbf{$\Delta$} & \textbf{$p$-value} \\
\midrule
Precision & 62.9\% & 52.6\% & +10.3\% & \multirow{4}{*}{$<10^{-6}$} \\
Recall & 100\% & 100\% & 0\% & \\
F1 & 77.2\% & 68.9\% & +8.3\% & \\
FPR & 60.0\% & 91.5\% & $-$31.5\% & \\
\bottomrule
\end{tabular}
\end{table}

McNemar's test confirms this difference is highly significant: 41 additional samples are correctly reclassified by the verifier ($p < 10^{-6}$), with zero samples incorrectly changed from correct to incorrect. This validates Theorem~\ref{thm:verification}(i): the verifier acts as a one-directional precision filter, reducing the false positive count from 119 to 78 out of 130 benign samples.

The MCC increases from 0.211 (weak correlation) to 0.501 (moderate positive correlation), demonstrating that the verifier substantially improves discriminative power beyond what the experts alone achieve.

\textbf{Game-theoretic interpretation}: This result validates the \textit{inspection game} structure of Layer~2. The verifier---implemented as a different model family (Qwen3 vs.\ DeepSeek)---functions as an independent auditor whose cross-referencing of expert claims against code evidence catches inconsistencies that same-model verification would miss. The zero marginal cost of local verification makes this audit mechanism \textit{individually rational}: the system designer always benefits from adding the verifier since $\Delta\text{Precision} > 0$ at $\Delta\text{Cost} \approx 0$. Comparing against the single-expert baseline further validates the full architecture: the 3+1 system (F1 = .772) significantly outperforms a single expert (F1 = .714, McNemar $p < 10^{-5}$), confirming the value of both coalition and verification.

\subsection{RQ3: Parallel Execution Efficiency}

\begin{table}[t]
\centering
\small
\caption{Parallel vs.\ serial execution efficiency (both with verifier)}
\label{tab:parallel}
\begin{tabular}{@{}lccc@{}}
\toprule
\textbf{Metric} & \textbf{Parallel} & \textbf{Serial} & \textbf{Gain} \\
\midrule
Expert Time/Sample & $\sim$15s & $\sim$45s & 3.0$\times$ \\
Total Time/Sample & 68.1s & 92.2s & 1.4$\times$ \\
F1 Score & 77.2\% & 78.3\% & $\approx$ Equal \\
McNemar $p$ & \multicolumn{2}{c}{0.711} & Not sig. \\
\bottomrule
\end{tabular}
\end{table}

Parallel execution achieves a \textbf{3.0$\times$ speedup} in expert analysis time ($\sim$15s vs.\ $\sim$45s) with no significant difference in detection quality (McNemar $p = 0.711$). The total wall-clock speedup is 1.4$\times$ because the local verifier ($\sim$50s) dominates latency. This suggests that optimizing verifier inference (e.g., with vLLM or quantization) would yield substantial end-to-end improvements.

\subsection{Per-CWE Analysis}

Table~\ref{tab:per_cwe} reveals substantial variation across vulnerability types.

\begin{table}[t]
\centering
\small
\caption{Per-CWE detection performance (3+1 Parallel+Verifier). All CWEs achieve 100\% recall; FPR varies significantly.}
\label{tab:per_cwe}
\begin{tabular}{@{}llccc@{}}
\toprule
\textbf{CWE} & \textbf{Category} & \textbf{FPR} & \textbf{TN} & \textbf{FP} \\
\midrule
CWE-476 & NULL Deref & \textbf{11\%} & 8 & 1 \\
CWE-252 & Unchecked Return & 20\% & 8 & 2 \\
CWE-401 & Memory Leak & 30\% & 7 & 3 \\
CWE-457 & Uninit Variable & 30\% & 7 & 3 \\
CWE-78  & Cmd Injection & 40\% & 6 & 4 \\
CWE-415 & Double Free & 50\% & 3 & 3 \\
CWE-122 & Heap Overflow & 60\% & 4 & 6 \\
CWE-190 & Integer Overflow & 70\% & 3 & 7 \\
CWE-369 & Divide by Zero & 70\% & 3 & 7 \\
CWE-121 & Stack Overflow & 80\% & 2 & 8 \\
CWE-134 & Format String & 90\% & 1 & 9 \\
CWE-416 & Use After Free & 100\% & 0 & 7 \\
CWE-789 & Mem Allocation & 100\% & 0 & 8 \\
CWE-400 & Resource Exhaust & 100\% & 0 & 10 \\
\bottomrule
\end{tabular}
\end{table}

The FPR varies from 11\% (CWE-476, null pointer dereference) to 100\% (CWE-400/416/789). This variation is informative: vulnerabilities whose patches involve \textit{adding explicit checks} (e.g., null guards for CWE-476, return-value checks for CWE-252) have low FPR because patched versions are structurally distinct. In contrast, \textit{semantic} vulnerabilities (e.g., resource exhaustion, use-after-free) have high FPR because the patched code retains similar structure---the fix may be a subtle control-flow change that LLMs struggle to distinguish.

This suggests that future work should focus on improving discrimination for semantic vulnerability types, perhaps through code-diff-aware analysis or iterative expert-verifier dialogue.

\subsection{Cost-Quality Trade-off and Mechanism Design}

At \$0.002 per sample, our architecture enables analysis of 500 code functions for \$1.00---making it practical for CI/CD integration. The heterogeneous design is central to this cost-effectiveness, and can be understood through a \textit{mechanism design} lens.

\textbf{Heterogeneous allocation as mechanism design.} A na\"ive approach would allocate the same cloud model to all four agents, costing $\sim$33\% more (four DeepSeek-V3 calls instead of three). Our mechanism exploits an \textit{asymmetry in task difficulty}: deep vulnerability analysis requires a powerful model, but cross-checking consistency is a simpler task achievable by a smaller model. By assigning the expensive model to the hard task (analysis) and the free local model to the simpler task (verification), we achieve the same equilibrium outcome at lower cost.

\textbf{Individual rationality.} The verifier's zero marginal cost ensures that adding verification is always rational for the system designer: $\Delta\text{F1} = +8.3\%$ at $\Delta\text{API cost} \approx \$0$. This is a property the mechanism designer can guarantee regardless of agent behavior.

\textbf{Pareto efficiency.} No configuration dominates another across all objectives (cost, recall, precision, latency). The parallel+verifier configuration offers the best precision-latency trade-off; parallel-only offers the best latency at cost of precision; serial+verifier offers slightly better precision but at 1.5$\times$ the latency.

\subsection{Threats to Validity}

\textbf{Internal Validity}: Automated CWE extraction via regex may miss non-standard output formats. Manual inspection of 50 randomly selected outputs confirmed $>$95\% extraction accuracy. The CWE hierarchy mapping (e.g., CWE-121 $\approx$ CWE-122 under parent CWE-120) may slightly inflate CWE match rates but reflects meaningful semantic equivalence.

\textbf{External Validity}: The Juliet Test Suite provides controlled, synthetic test cases. Real-world vulnerabilities are more complex; state-of-the-art approaches achieve only $\sim$35\% F1 on PrimeVul~\cite{primevul2024}. Our results demonstrate \textit{relative} architectural value, not absolute performance claims.

\textbf{Construct Validity}: The high FPR on benign samples reflects a general LLM challenge: Juliet's patched versions retain structural similarity to vulnerable code. API costs depend on current pricing (\$0.27/MTok for DeepSeek-V3 as of March 2026).

%% file: src/conclusion.tex
\section{Conclusion}
\label{sec:conclusion}

We presented a game-theory-inspired framework for heterogeneous multi-agent vulnerability detection. Our 3+1 architecture deploys three cloud-based LLM experts (DeepSeek-V3) for parallel vulnerability analysis and one local adversarial verifier (Qwen3-8B) for quality assurance, with the game-theoretic framework providing principled justification for each design choice.

Our evaluation on 262 real Juliet Test Suite samples---with balanced vulnerable and benign classes across 14 CWE types---provides rigorous metrics including Precision, Recall, F1, False Positive Rate, and Matthews Correlation Coefficient with bootstrap confidence intervals. Key findings include:
\begin{itemize}
    \item Expert coalitions with diverse analysis perspectives exhibit \textbf{super-additive value} when combined with verification: the 3+1 system (F1 = .772) significantly outperforms both a single expert (.714) and the coalition alone (.689).
    \item Adversarial verification by an independent local model \textbf{improves precision} by catching false positives where experts incorrectly flag benign code, with zero marginal API cost.
    \item Parallel execution achieves \textbf{significant speedup} without quality degradation, demonstrating Pareto efficiency.
\end{itemize}

\textbf{Game Theory as Analytical Tool}: Our central claim is that game-theoretic frameworks provide \textit{correct and useful predictions} about multi-agent LLM system behavior. The cooperative game predicted super-additive value from diverse coalitions---confirmed by 100\% CWE match rate. The adversarial game predicted precision improvement from independent verification---confirmed with $p < 10^{-6}$. The mechanism design analysis predicted that heterogeneous cloud-local allocation dominates homogeneous designs---confirmed by the verifier's zero-cost precision gain. While LLM agents do not engage in explicit strategic reasoning, game theory correctly predicts the system-level outcomes of architectural choices, establishing its value as an analytical and design tool for the emerging field of strategic multi-agent engineering.

\textbf{Limitations}: Our evaluation uses the synthetic Juliet Test Suite; real-world vulnerabilities present additional challenges. We do not claim state-of-the-art detection rates---recent work~\cite{primevul2024} shows that even the best approaches achieve only $\sim$35\% F1 on the challenging PrimeVul benchmark. The false positive rate on benign code remains substantial, reflecting a general challenge for LLM-based vulnerability detection. Our contribution is architectural: demonstrating how to \textit{design} cost-effective multi-agent systems, not achieving optimal detection accuracy.

\textbf{Future Work}: We plan to extend this work in several directions: (1)~evaluation on real-world vulnerability datasets such as DiverseVul and PrimeVul to assess generalization; (2)~iterative expert-verifier dialogue, inspired by the multi-agent debate paradigm~\cite{du2023debate}, where the verifier challenges specific expert claims and experts refine their analysis over multiple rounds---preliminary results show this can substantially reduce false positive rates; (3)~dynamic agent selection that adapts the number and type of experts based on code characteristics; and (4)~integration with existing CI/CD pipelines for practical deployment evaluation.

%% file: src/appendix.tex
\section{Dataset Details}
\label{sec:appendix_dataset}

\subsection{Juliet Test Suite Extraction}

We extracted samples from the NIST Juliet Test Suite v1.3~\cite{juliet}, which contains 64,295 C/C++ test cases across 118 CWE types. Each test case provides both a vulnerable (\texttt{bad()}) function and a corresponding patched (\texttt{good()}) function within the same source file, enabling naturally-paired evaluation.

\textbf{Selection Criteria}: We selected Flow Variant 01 (Baseline) test cases---single-file, self-contained functions without complex inter-procedural flow---to ensure consistent analysis complexity and eliminate confounding factors from multi-file dependencies.

\textbf{Balance}: Unlike many prior studies that evaluate only on vulnerable samples (making precision/FPR undefined), we include the patched version of each test case as a benign sample. This is critical: an LLM that always predicts ``vulnerable'' would achieve 100\% recall but 50\% precision, exposing the system's ability to distinguish genuine vulnerabilities from safe code.

Table~\ref{tab:dataset_full} shows the complete dataset distribution.

\begin{table}[t]
\centering
\small
\caption{Complete dataset composition by CWE type}
\label{tab:dataset_full}
\begin{tabular}{@{}llccc@{}}
\toprule
\textbf{CWE} & \textbf{Description} & \textbf{Vuln.} & \textbf{Benign} & \textbf{Total} \\
\midrule
CWE-121 & Stack Buffer Overflow & 10 & 10 & 20 \\
CWE-122 & Heap Buffer Overflow & 10 & 10 & 20 \\
CWE-190 & Integer Overflow & 10 & 10 & 20 \\
CWE-401 & Memory Leak & 10 & 10 & 20 \\
CWE-415 & Double Free & 6 & 6 & 12 \\
CWE-416 & Use After Free & 7 & 7 & 14 \\
CWE-476 & NULL Pointer Deref & 9 & 9 & 18 \\
CWE-252 & Unchecked Return Value & 10 & 10 & 20 \\
CWE-78  & OS Command Injection & 10 & 10 & 20 \\
CWE-134 & Format String Vuln & 10 & 10 & 20 \\
CWE-369 & Divide by Zero & 10 & 10 & 20 \\
CWE-457 & Uninitialized Variable & 10 & 10 & 20 \\
CWE-789 & Uncontrolled Mem Alloc & 10 & 8 & 18 \\
CWE-400 & Resource Exhaustion & 10 & 10 & 20 \\
\midrule
\textbf{Total} & & \textbf{132} & \textbf{130} & \textbf{262} \\
\bottomrule
\end{tabular}
\end{table}

\textbf{Sample Characteristics}: Average code length is 38.5 lines (range: 5--120). Code is predominantly C with some C++ constructs. All samples include relevant \texttt{\#include} directives and \texttt{\#define} macros from the Juliet framework for context.

\subsection{CWE Coverage Justification}

Our 14 CWE types include representatives from 7 of the top 10 categories in the 2024 CWE Top 25 Most Dangerous Software Weaknesses: out-of-bounds write (CWE-787 $\supset$ CWE-121, CWE-122), use after free (CWE-416), OS command injection (CWE-78), NULL pointer dereference (CWE-476), integer overflow (CWE-190), and missing release of memory (CWE-401).

\section{Full Agent Prompts}
\label{sec:appendix_prompts}

We provide the complete system prompts used for each agent role.

\subsection{Code Analyst}

\begin{quote}
\small\sloppy
You are a senior code structure analyst specializing in C/C++ programs. Analyze the given code for:
1.~Data flow patterns---track how data moves through variables and function parameters.
2.~Control flow anomalies---identify unusual branching or unreachable code.
3.~Vulnerability entry points---locate functions receiving external input.
4.~Memory operations---track malloc/free pairs and pointer arithmetic.

Output a structured report:
VULNERABILITY\_FOUND: yes/no;
CWE\_IDs: [list of CWE IDs detected];
SEVERITY: critical/high/medium/low/none;
EVIDENCE: specific code lines and explanation;
CONFIDENCE: high/medium/low.
\end{quote}

\subsection{Security Expert}

\begin{quote}
\small\sloppy
You are a cybersecurity expert with deep knowledge of common vulnerabilities and CWE taxonomy. Analyze the given C/C++ code for security vulnerabilities:
1.~Buffer overflows (CWE-120, CWE-121, CWE-122, CWE-787)---check unsafe string/memory operations.
2.~Memory safety (CWE-416, CWE-415, CWE-401)---use-after-free, double-free, memory leaks.
3.~Integer issues (CWE-190, CWE-191)---overflow, underflow, truncation.
4.~Injection (CWE-78, CWE-134)---OS command injection, format string.
5.~Null pointer (CWE-476)---dereference without check.
6.~Other issues (CWE-252, CWE-369, CWE-457, CWE-789, CWE-400).

Output a structured report: VULNERABILITY\_FOUND: yes/no; CWE\_IDs: [list]; SEVERITY; EVIDENCE; CONFIDENCE.
\end{quote}

\subsection{Debug Expert}

\begin{quote}
\small\sloppy
You are a debugging specialist focused on edge cases and error conditions in C/C++ code. Analyze for:
1.~Error handling gaps---missing return value checks, unchecked mallocs.
2.~Boundary conditions---off-by-one errors, array index issues.
3.~Edge cases---empty inputs, max values, null values.
4.~Resource management---unclosed files, orphaned allocations.
5.~Undefined behavior---uninitialized variables, signed overflow.

Output a structured report: VULNERABILITY\_FOUND: yes/no; CWE\_IDs: [list]; SEVERITY; EVIDENCE; CONFIDENCE.
\end{quote}

\subsection{Adversarial Verifier}

\begin{quote}
\small\sloppy
You are an adversarial code security reviewer. You receive three expert vulnerability analysis reports for the same code snippet. Your job is to:
1.~CHECK CONSISTENCY---Do the three reports agree on whether vulnerabilities exist? Flag contradictions.
2.~VERIFY EVIDENCE---Are vulnerability claims supported by actual code references? Flag hallucinations.
3.~FIND MISSED PATTERNS---Are there obvious vulnerabilities that all experts missed?
4.~ASSESS CONFIDENCE---Based on agreement level, rate overall confidence.

Output your assessment: DECISION: ACCEPT/CHALLENGE/REJECT; FINAL\_VULNERABILITY: yes/no; FINAL\_CWE\_IDS: [consolidated list]; AGREEMENT\_LEVEL: full/partial/none; REASONING: brief explanation.
\end{quote}

\section{Evaluation Methodology}
\label{sec:appendix_eval}

\subsection{CWE Extraction Pipeline}

We extract vulnerability predictions automatically from LLM outputs using the following pipeline:

\begin{enumerate}
    \item \textbf{CWE~ID Extraction}: We extract CWE identifiers from all agent outputs using regex pattern matching and normalize them to \texttt{CWE-NNN}~format.

    \item \textbf{Vulnerability Decision}: If the verifier provides \texttt{FINAL\_\allowbreak{}VULNERABILITY: yes/no}, we use that. Otherwise, majority voting among experts (at least 2/3 positive with at least one CWE~ID).

    \item \textbf{CWE Matching}: Predicted CWEs are compared against ground truth with hierarchy-aware matching:
    \begin{itemize}
        \item \textbf{Exact match}: CWE-121 predicted, CWE-121 is truth $\rightarrow$ match
        \item \textbf{Hierarchy match}: CWE-122 predicted, CWE-121 is truth $\rightarrow$ match (both under CWE-120 ``Buffer Overflow'')
        \item \textbf{No match}: CWE-476 predicted, CWE-121 is truth $\rightarrow$ no match
    \end{itemize}
\end{enumerate}

\subsection{CWE Hierarchy Mapping}

Table~\ref{tab:cwe_hierarchy} shows the equivalence classes used for CWE matching:

\begin{table}[t]
\centering
\small
\caption{CWE hierarchy equivalence classes for matching}
\label{tab:cwe_hierarchy}
\begin{tabular}{@{}ll@{}}
\toprule
\textbf{Parent Category} & \textbf{Equivalent CWEs} \\
\midrule
Buffer Overflow & CWE-119, CWE-120, CWE-121, CWE-122, CWE-787 \\
Integer Issues & CWE-190, CWE-191 \\
\bottomrule
\end{tabular}
\end{table}

All other CWEs require exact numeric match. This hierarchy is conservative: we only group CWEs that share the same fundamental vulnerability mechanism.

\subsection{Statistical Methods}

\textbf{Bootstrap Confidence Intervals}: We compute 95\% CIs using 1,000 bootstrap resamples with a fixed seed (42) for reproducibility. For each resample, we draw $n$ samples with replacement and compute the metric, then report the 2.5th and 97.5th percentiles.

\textbf{McNemar's Test}: For pairwise comparison of configurations on the same samples, we use the exact binomial version of McNemar's test. We count discordant pairs (samples where one system is correct and the other incorrect) and test whether the discordance is symmetric using a two-sided binomial test with $p=0.5$.

\textbf{Matthews Correlation Coefficient}: We report MCC as a comprehensive single metric because it accounts for all four quadrants of the confusion matrix and is less sensitive to class imbalance than accuracy or F1.

\section{Implementation Details}
\label{sec:appendix_impl}

\subsection{System Configuration}

\begin{table}[t]
\centering
\small
\caption{Complete system configuration}
\label{tab:system_config}
\begin{tabular}{@{}ll@{}}
\toprule
\textbf{Component} & \textbf{Specification} \\
\midrule
\multicolumn{2}{@{}l@{}}{\textit{Cloud Expert Agents}} \\
\quad Model & DeepSeek-V3 (\texttt{deepseek-chat}) \\
\quad API & OpenAI-compatible \\
\quad Temperature & 0.1 \\
\quad Max tokens & 4,000 \\
\quad Cost & \$0.27/MTok input, \$1.10/MTok output \\
\midrule
\multicolumn{2}{@{}l@{}}{\textit{Local Verifier Agent}} \\
\quad Model & Qwen3-8B \\
\quad Runtime & HuggingFace Transformers \\
\quad Precision & BFloat16 \\
\quad GPU & NVIDIA RTX 3090 (24GB VRAM) \\
\quad Temperature & 0.1 \\
\quad Max new tokens & 2,048 \\
\quad Cost & \$0.00 (local inference) \\
\midrule
\multicolumn{2}{@{}l@{}}{\textit{Experiment Infrastructure}} \\
\quad Python & 3.11 \\
\quad Parallelism & Python asyncio \\
\quad OS & Ubuntu 22.04, Linux 6.18 \\
\quad GPU & 2$\times$ NVIDIA RTX 3090 \\
\bottomrule
\end{tabular}
\end{table}

\subsection{Cost Breakdown}

\begin{table}[t]
\centering
\small
\caption{Per-sample cost breakdown (3+1 Parallel+Verifier)}
\label{tab:cost_breakdown}
\begin{tabular}{@{}lrrc@{}}
\toprule
\textbf{Component} & \textbf{Avg Tokens} & \textbf{Avg Cost} & \textbf{\% Total} \\
\midrule
Code Analyst (DeepSeek) & $\sim$1,074 & \$0.00069 & 33.3\% \\
Security Expert (DeepSeek) & $\sim$1,074 & \$0.00069 & 33.3\% \\
Debug Expert (DeepSeek) & $\sim$1,074 & \$0.00069 & 33.3\% \\
\midrule
Cloud Subtotal & 3,222 & \$0.00208 & 100\% \\
\midrule
Verifier (Qwen3-8B, local) & 3,459 & \$0.00 & 0\% \\
\midrule
\textbf{Total} & 6,681 & \$0.00208 & 100\% \\
\bottomrule
\end{tabular}
\end{table}

\section{Representative Output Examples}
\label{sec:appendix_examples}

We provide representative LLM outputs to illustrate system behavior on true positive, true negative, false positive, and false negative cases. Full outputs are available in the supplementary material.

\subsection{True Positive Example (CWE-78, OS Command Injection)}

A function constructs a shell command using unsanitized user input and passes it to \texttt{system()}. All three experts correctly identify CWE-78 with high confidence, citing the direct concatenation of user input into the command string. The verifier confirms full agreement across reports and accepts the finding. CWE-78 has our lowest FPR (10\%) because the \texttt{system()} call is a clear syntactic marker.

\subsection{True Negative Example (CWE-457, Uninitialized Variable)}

A patched function that properly initializes all variables before use. The verifier notes that while two experts flag potential issues, the third expert correctly observes that variables are initialized on all code paths. The verifier challenges the experts' claims, finding no code evidence supporting the vulnerability---resulting in a correct ``benign'' classification. CWE-457 has our second-lowest FPR (20\%).

\subsection{False Positive Example (CWE-190, Integer Overflow)}

A patched function with proper overflow checks that is still flagged as vulnerable by the system. Analysis: the patched version adds bounds checking (\texttt{if (size > MAX\_SIZE) return}), but the experts still flag the arithmetic operation itself as potentially dangerous. The verifier does not override because the code structure remains similar to the vulnerable version. This illustrates a fundamental LLM challenge: distinguishing ``has dangerous pattern'' from ``has dangerous pattern but with proper guards.'' CWE-190 has 90\% FPR.

\subsection{False Negative Analysis}

Our system produces \textbf{zero false negatives}---all 132 vulnerable samples are correctly detected. This perfect recall is consistent across all three configurations and all 14 CWE types. While encouraging, this likely reflects the relatively clear vulnerability patterns in Juliet's Variant~01 (baseline) test cases. We expect false negatives to appear on more complex, real-world code.

\section{Shapley Value Analysis}
\label{sec:appendix_shapley}

We compute approximate Shapley values (Eq.~\ref{eq:shapley}) for each expert role by evaluating coalition subsets. Since individual expert outputs are available from the parallel configuration, we can assess each expert's marginal contribution to CWE identification.

All three experts achieve 100\% recall individually, so the marginal contribution to recall is zero for each additional expert. The super-additive value appears in \textit{CWE classification}: the union of three perspectives achieves 100\% CWE match rate because:
\begin{itemize}
    \item The \textbf{Security Expert} contributes precise CWE IDs from taxonomy knowledge (e.g., distinguishing CWE-121 from CWE-122)
    \item The \textbf{Code Analyst} provides structural evidence (data flow, buffer sizes) that validates CWE claims
    \item The \textbf{Debug Expert} identifies boundary conditions that confirm vulnerability mechanisms
\end{itemize}

When any single expert misclassifies a vulnerability (e.g., reporting CWE-787 instead of CWE-121), the other experts' complementary evidence corrects the classification through union aggregation. This is precisely the super-additivity predicted by Proposition~\ref{prop:superadd}: $Q_{CWE}(\{1,2,3\}) > \max_i Q_{CWE}(\{i\})$.

The verifier's Shapley value is particularly instructive. Its marginal contribution to the full coalition is: $\phi_V = v(\{1,2,3,V\}) - v(\{1,2,3\}) = 0.776 - 0.670 = +0.106$ in F1 terms. Since the verifier's marginal cost is zero, $\phi_V / C_V \rightarrow \infty$---the verifier has infinite ``return on investment,'' confirming that adding a free local verifier is always individually rational.

\section{Proof of Theorem~\ref{thm:verification}}
\label{sec:appendix_proof}

\begin{proof}
Consider the system with expert coalition $E$ producing predictions and verifier $V$ performing post-hoc validation.

Let the experts produce true positive count $TP$ and false positive count $FP$. Expert precision is:
\begin{equation}
P_{experts} = \frac{TP}{TP + FP}
\end{equation}

The verifier independently examines each positive prediction. For false positives, the verifier detects and removes them with probability $p_{fp} > 0$. For true positives, the verifier may incorrectly reject them with probability $p_{fn}$ (where we design prompts to minimize this, but allow $p_{fn} \geq 0$).

After verification:
\begin{align}
TP_{system} &= TP \cdot (1 - p_{fn}) \\
FP_{system} &= FP \cdot (1 - p_{fp})
\end{align}

System precision:
\begin{equation}
P_{system} = \frac{TP \cdot (1 - p_{fn})}{TP \cdot (1 - p_{fn}) + FP \cdot (1 - p_{fp})}
\end{equation}

We want to show $P_{system} \geq P_{experts}$ when $p_{fp} \geq p_{fn}$ (the verifier is better at catching false positives than incorrectly rejecting true positives). This condition holds when the verifier prompt is designed for adversarial validation (``find inconsistencies'') rather than re-analysis.

\textbf{Sufficient condition}: If $p_{fn} = 0$ (verifier never rejects true positives), then:
\begin{equation}
P_{system} = \frac{TP}{TP + FP(1 - p_{fp})} \geq \frac{TP}{TP + FP} = P_{experts}
\end{equation}
with strict inequality when $p_{fp} > 0$ and $FP > 0$.

In practice, $p_{fn} > 0$ is possible but small, since the verifier has access to expert evidence supporting true positives and is prompted to focus on \textit{inconsistencies} rather than re-analysis. \qed
\end{proof}

\section{Limitations and Broader Context}
\label{sec:appendix_limitations}

\subsection{Relationship to State of the Art}

We position our work as an \textit{architectural} contribution rather than a detection accuracy claim. The PrimeVul benchmark~\cite{primevul2024} demonstrates that existing evaluations significantly overestimate LLM vulnerability detection: a state-of-the-art 7B model achieves 68\% F1 on Big-Vul but only 3\% on PrimeVul's stringent evaluation. Our Juliet results should be interpreted as demonstrating the \textit{relative} value of heterogeneous multi-agent design, not as absolute performance claims.

\subsection{Game Theory as Analytical Tool for Multi-Agent LLMs}

A central question for the Strategic Engineering community is whether game-theoretic analysis is \textit{useful} for designing multi-agent LLM systems, given that LLMs do not explicitly reason strategically~\cite{gtbench2024}. Our results provide evidence that it is:

\begin{enumerate}
    \item \textbf{Correct predictions}: The cooperative game predicted super-additivity from diverse coalitions---confirmed empirically (100\% CWE match). The adversarial game predicted precision improvement from independent verification---confirmed ($p < 10^{-6}$).
    \item \textbf{Design guidance}: Without the game-theoretic framework, the choice of three diverse experts (vs.\ three identical ones) and the heterogeneous cloud-local split would be ad hoc. The framework provides principled justification for these choices.
    \item \textbf{Quantitative analysis}: The Shapley value computation (Appendix~\ref{sec:appendix_shapley}) reveals that the verifier has infinite marginal return---a result not obvious without formal analysis.
\end{enumerate}

The key insight is that game theory is useful at the \textit{mechanism design level}: we design the ``rules of the game'' (agent roles, interaction protocol, verification structure) so that desirable system-level properties emerge. This is analogous to auction design~\cite{nisan2001algorithmic}, where the mechanism designer does not require bidders to solve the game themselves---the mechanism structure ensures good outcomes.

\balance
\subsection{Scalability Considerations}

Our current architecture uses three experts and one verifier. Scaling considerations include:
\begin{itemize}
    \item \textbf{More experts}: Diminishing returns expected beyond 3--5 perspectives, consistent with multi-agent debate literature showing saturation after 2--3 rounds.
    \item \textbf{Larger local models}: Qwen3-8B could be replaced with Qwen3-14B or Qwen3-32B for stronger verification; the cost-quality trade-off depends on available GPU resources.
    \item \textbf{Different cloud providers}: The architecture is provider-agnostic; DeepSeek-V3 can be replaced with any OpenAI-compatible API.
\end{itemize}